\documentclass[12pt]{article}%
\usepackage{graphicx}%
\usepackage[normalem]{ulem}
\usepackage{amsmath,amsthm,amsfonts,
amsbsy,amssymb,upref,enumerate,bigstrut,
color,mathtools,mathrsfs,float,bm,dsfont,scalefnt,mathtools}
\usepackage{natbib}
\usepackage{lipsum}
\usepackage{stmaryrd}
\usepackage{authblk} 
\usepackage{subfigure}
\usepackage{pgfplots}
\usetikzlibrary{pgfplots.fillbetween}

\usepackage{setspace} 

\usepackage{newtxtext,newtxmath} 
\usepackage{mhchem} 

\usepackage{caption}
\usepackage{tikz}
\usetikzlibrary{intersections}

\usepackage[colorlinks=true,linkcolor=blue,citecolor=blue]{hyperref}

\usepackage[left=1in,top=1.2in,right=0.7in]{geometry}

\newtheorem{theorem}{Theorem}[section]

\newtheorem{corollary}[theorem]{Corollary}

\newtheorem{proposition}[theorem]{Proposition}
\newtheorem{remark}[theorem]{Remark}
\numberwithin{equation}{section} 

\makeatletter
\def\@seccntformat#1{\@ifundefined{#1@cntformat}%
	{\csname the#1\endcsname\quad}
	{\csname #1@cntformat\endcsname}
}
\makeatother

\newif\ifShowComments
\ShowCommentstrue
\def\strutdepth{\dp\strutbox}
\def\druk#1{\strut\vadjust{\kern-\strutdepth
        {\vtop to \strutdepth{%
                \baselineskip\strutdepth\vss
                        \llap{\hbox{#1}\quad}\null}}}}

\title{\bf 
On the bias of the Gini estimator: Poisson and geometric cases, a characterization of the gamma family, and unbiasedness under gamma distributions
}

\author[1]{Roberto Vila
	\thanks{rovig161@gmail.com. \, (Corresponding author)}
	}
\author[1,2]{Helton Saulo\thanks{heltonsaulo@gmail.com}}
\affil[1]{Department of Statistics, University of
	 Bras\'ilia, Bras\'ilia, Brazil}
\affil[2]{
Department of Economics, Federal University of Pelotas, Pelotas, Brazil}

\begin{document}
\maketitle

\begin{abstract}
In this paper, we derive a general representation for the expectation of the Gini coefficient estimator in terms of the Laplace transform of the underlying distribution, together with the mean and the Gini coefficient of its exponentially tilted version. This representation leads to a new characterization of the gamma family within the class of nonnegative scale families, based on a stability property under exponential tilting. As direct applications, we show that the Gini estimator is biased for both Poisson and geometric populations and provide an alternative, unified proof of its unbiasedness for gamma populations. By using the derived bias expressions, we propose plug-in bias-corrected estimators and assess their finite-sample performance through a Monte Carlo study, which demonstrates substantial improvements over the original estimator. Compared with existing approaches, our framework highlights the fundamental role of scale invariance and exponential tilting, rather than distribution-specific algebraic calculations, and complements recent results in \citet{Baydil2025} and \citet{VilaSaulomixture2025,VilaSaulomth2025}.
\end{abstract}


\section{Introduction}

The Gini coefficient is a widely used measure of inequality in economics, finance, demography, and reliability theory, and has been extensively studied since its introduction by \citet{Gini1936}, particularly with regard to the statistical properties of its estimators.

A well-known issue is the finite-sample bias of the classical Gini estimator. To address it, \citet{Deltas2003} proposed an upward-adjusted estimator that corrects small-sample bias under particular distributional assumptions. Its unbiasedness has since been established for several distribution families, most notably the gamma family; see \citet{Baydil2025} and \citet{VilaSaulomixture2025,VilaSaulomth2025}. These results raise a natural question: what structural properties of the underlying distribution explain the unbiasedness of the Gini estimator, and to what extent does this phenomenon extend beyond the gamma family?

This paper addresses these questions by developing a general representation for the expectation of the Gini estimator that applies to any non-degenerate nonnegative distribution with finite mean. The proposed representation expresses the expectation of the estimator in terms of the Laplace transform of the underlying distribution and the Gini coefficient of its exponentially tilted version. This formulation naturally connects with the Gini mean difference framework studied in \citet{VILA2024110032} and provides a transparent way to analyze bias properties across different distributional families.

Within this framework, we obtain four main contributions. First, we show explicitly that the Gini estimator is biased for Poisson populations, deriving an exact integral representation for its expectation and sharp bounds for the resulting bias. Second, we extend this analysis to geometric populations, where we again establish bias and provide closed-form and hypergeometric representations for the expectation of the estimator. These results demonstrate that bias persists even for simple, classical discrete distributions. Third, we show that a scale-invariance property of exponentially tilted distributions characterizes the gamma family, yielding a unified and conceptually simple proof of the unbiasedness of the Gini estimator under gamma populations. Fourth, we use the derived bias expressions of the Gini estimators (Poisson and geometric) to propose bias-corrected versions through a plug-in approach. A Monte Carlo study is carried out to assess the finite-sample behavior of both the original and the corrected estimators. The simulation results indicates that the bias-corrected estimators exhibit a markedly improved performance.

The remainder of the paper is organized as follows. In Section \ref{sec:02}, we review basic properties and representations of the Gini coefficient. In Section \ref{Deriving estimator biases}, we derive a general expression for the expectation of the Gini estimator.
In Sections \ref{Poisson-populations} and \ref{Geometric-populations}, we apply this result to Poisson and geometric populations, respectively, establishing finite-sample bias in both cases.
In Section \ref{Characterization}, we establish a characterization of the gamma family based on exponential tilting, and Section \ref{Unbiasedness} uses this characterization to obtain an alternative proof of the unbiasedness of the Gini estimator for gamma populations. In Section~\ref{sec:sim}, we propose bias-corrected estimators for the Poisson and geometric cases and report results from a simulation study assessing the finite-sample performance of all estimators. Finally, in Section \ref{concluding_remarks}, we conclude with a discussion of implications and directions for future research.

\section{The Gini coefficient}\label{sec:02}

Let \(X_1, X_2\) be independent and identically distributed (i.i.d.) 
random variables with common distribution as a non-degenerate nonnegative
random variable $X$, and finite mean \(\mu=\mathbb{E}[X]>0\).
The Gini coefficient \citep{Gini1936} is defined as
\begin{equation*}
	G \equiv G(X)=\frac{\mathbb{E}\lvert X_1-X_2\rvert}{2\mu}.
\end{equation*}

The Gini coefficient possesses the following well-known fundamental properties and representations:

\begin{itemize}
	\item[(G1)] \textit{Existence:} the Gini coefficient exists and satisfies \(0\leqslant G<1\).
	
	\item[(G2)] \textit{Scale invariance:} for any \(b>0\),
	$
	G(bX)=G(X).
	$
	
	\item[(G3)] \textit{Lack of translation invariance:} for any \(a>0\),
	$
	G(a+X)={\mu}G(X)/(a+\mu).
	$
	\item[(G4)] \textit{Covariance and quantile representations:}
	$
	G
	=
	\mathrm{Cov}\!\left(X,\,2F_X(X)-1\right)/\mu
	=
	\int_0^1 F_X^{-1}(p)(2p-1)\,\mathrm{d}p,
	$
	where $F_X$ is the cumulative distribution function (CDF) of $X$ and $F_X^{-1}$ its generalized inverse function.
	
	\item[(G5)] \textit{Decomposition via Lorenz measures of inequality:} 
	$
	G=[D_1(F_X)+G_2(F_X)]/2,
	$
	where
	$D_n(F_X)=(n+1)\mathbb{E}\left[(U-L(U))U^{\,n-1}\right]$ and 
	$G_n(F_X)=n(n-1)\mathbb{E}\left[(U-L(U))(1-U)^{\,n-2}\right]$, $n\geqslant 1$,
	are the Lorenz measure of inequality  \citep{Aaberge2000} and
	the generalized Gini measure 
	\citep{Donaldson1980,Kakwani1980,Yitzhaki1983}, respectively,
	with \(U\sim\mathrm{U}(0,1)\), and
	$
		L(p)=\int_0^p F_X^{-1}(t)\,\mathrm{d}t/\mu,
		\, 0\leqslant p\leqslant 1,
	$
	being the Lorenz curve associated with $X$.
\end{itemize}

The following result is well known in the literature; for completeness and the reader's convenience, we briefly outline the main idea of its proof.
\begin{proposition}\label{chave-prop}
The Gini coefficient $G$ of a nonnegative
random variable $X$ can be characterized as
\begin{align*}
	G={\int_{0}^{\infty}F_X(x^-)[1-F_X(x^-)]{\rm d}x\over \mu},
\end{align*}
where $\mu=\mathbb{E}[X]>0$ is the mean of $X$, $F_X$ is the CDF of $X$, and $F_X(x^-)=\mathbb{P}(X<x)$.
\end{proposition}
\begin{proof}
The proof is immediate upon using the basic identity:
\begin{align}\label{basic-identity}
	\vert x-y\vert
	=
	\int_0^{\infty}
	\left[
	\mathds{1}_{(x,\infty)}(t)
	+
	\mathds{1}_{(y,\infty)}(t)
	-
	2
	\mathds{1}_{(x,\infty)}(t)
		\mathds{1}_{(y,\infty)}(t)
	\right]
	{\rm d}t,
\end{align}
where $\mathds{1}_{A}(t)$ is the indicator function of a set $A$.
\end{proof}

\begin{proposition}\label{pro-2}
The Gini coefficient $G$ of a discrete
random variable $X\in\{0,1,\ldots\}$ can be characterized as
\begin{align*}
	G={\sum_{k=0}^\infty F_X(k)[1-F_X(k)]\over \mu}.
\end{align*}
\end{proposition}
\begin{proof}
	Since
	\[
	F_X(x^-)=F_X(\lfloor x\rfloor)=F_X(k), \quad x\in[k,k+1),
	\]
	where $\lfloor x\rfloor$ denotes the floor function, the function $F_X(x^-)$
	is constant on each interval $[k,k+1)$. Consequently,
	\[
	\int_{0}^{\infty} F_X(x^-)\,[1-F_X(x^-)]\,\mathrm{d}x
	=
	\sum_{k=0}^{\infty}\int_{k}^{k+1}
	F_X(\lfloor x\rfloor)\,[1-F_X(\lfloor x\rfloor)]\,\mathrm{d}x .
	\]
	Therefore, by Proposition~\ref{chave-prop}, the proof follows.
\end{proof}

	\section{A general expression for the expectation of the Gini  estimator}\label{Deriving estimator biases}

In this section (Theorem~\ref{main-theo}), we derive a general expression for the expectation of the Gini coefficient estimator $\widehat{G}$.

Let $X_1,\ldots,X_n$ be i.i.d. copies of $X$. Following  \cite{Pena2025}, we define an estimator of the Gini coefficient $G$ by
\begin{align}\label{def-G-estimator}
	\widehat{G}
	=
	\begin{cases}
	\dfrac{ {1\over \binom{n}{2}} \sum_{1\leqslant i<j\leqslant n}\vert X_i-X_j\vert}{2\overline{X} },
& \text{if} \  \sum_{k=1}^{n}X_k>0 \ \text{a.s.},
\\[0,3cm]
0, & \text{if} \  \sum_{k=1}^{n}X_k=0 \ \text{a.s.},
	\end{cases}
\end{align}
where
$
\overline{X}=\sum_{k=1}^{n} X_k / n
$
denotes the sample mean. Note that, in the case where
\(\sum_{k=1}^{n} X_k > 0\) almost surely, \(\widehat{G}\) coincides with the
upward adjusted Gini coefficient estimator originally proposed by
\citet{Deltas2003}.

\begin{remark}
In an earlier version of this work, the following result was established only for i.i.d. absolutely continuous random variables, independently of the results obtained in \cite{Pena2025}. Becoming aware of that reference motivated us to extend this result to i.i.d. random variables in a more general setting, which can be further extended, following the approach of \cite{Pena2025}, by removing the assumption that the variables involved are identically distributed.
\end{remark}

\begin{theorem}\label{main-theo}
Let \(X_1,\ldots,X_n\) be independent copies of a nonnegative
random variable \(X\) with finite mean $\mu=\mathbb{E}[X]>0$ and CDF \(F_X\).
Then
\[
\mathbb{E}\left[\widehat{G}\right]
=
n
\int_{0}^{\infty}
\mathbb{E}[Y_z]
G(Y_z)
\mathscr{L}_X^{\,n}(z)\,
\mathrm{d}z,
\]
where \(\mathscr{L}_X(z)=\mathbb{E}\left[\exp\{-zX\}\right]\) denotes the Laplace transform of \(X\).
For each \(z>0\), \(G(Y_z)\) is the Gini coefficient of 
an exponentially tilted (or Esscher-transformed) random variable
\(Y_z\) \citep{Butler2007}, whose CDF is given by
\begin{align}\label{notation}
	F_{Y_z}(t)
	=
	\frac{\mathbb{E}\left[\exp\{-zX\} \mathds{1}_{(0,t]}(X)\right]}
	{\mathscr{L}_X(z)},
	\quad t>0.
\end{align}
In the above, we implicitly assume that the Lebesgue–Stieltjes and improper integrals involved exist.
\end{theorem}
\begin{proof}
	Using the identity 		$
	\int_{0}^\infty \exp(-w z){\rm d}z={1/w},
	\, w>0,
	$ together with Tonelli's theorem, we can express $\mathbb{E}\left[\widehat{G}\right]$ as follows
\begin{align*}
	\mathbb{E}\left[\widehat{G}\right]
	&=
	{n\over 2\binom{n}{2}}
	\sum_{1\leqslant i<j\leqslant n}
	\int_{0}^{\infty}
	\mathbb{E}\left[\vert X_i-X_j\vert \exp\{-z(X_1+X_2)\}\exp\left\{-z\sum_{k=3}^{n}X_k\right\}\right]
	{\rm d}z
	\\[0,2cm]
	&=
	{n\over 2}
		\int_{0}^{\infty}
	\mathbb{E}\left[\vert X_1-X_2\vert \exp\{-z(X_1+X_2)\}\right]
	\mathscr{L}_X^{n-2}(z)
	{\rm d}z,
\end{align*}
where, in the last equality, we used the fact that $X_1,\ldots,X_n$
are i.i.d. copies of $X$.
Applying the identity in \eqref{basic-identity}, the above identity can be written as
%
\begin{align*}
\mathbb{E}\left[\widehat{G}\right]
	=
	{n}
\int_{0}^{\infty}
\left\{
\int_0^{\infty}
\mathbb{E}
\left[
\exp\{-zX\}
\mathds{1}_{(0,t)}(X)	
\right]
\left\{
\mathscr{L}_X(z)
-
\mathbb{E}
\left[
\exp\{-zX\}
\mathds{1}_{(0,t)}(X)	
\right]
\right\}
{\rm d}t \, 
\mathscr{L}_X^{n-2}(z)
\right\}
{\rm d}z.
\end{align*}

Using the notation of $F_{Y_z}$ in \eqref{notation}, the above identity becomes
\begin{align}\label{charact-1}
\mathbb{E}\left[\widehat{G}\right]
	=
	{n}
	\int_{0}^{\infty}
	\left\{
	\int_0^{\infty}
	F_{Y_z}(t^-) 
	\left[
	1
	-
	F_{Y_z}(t^-)
	\right] 
	{\rm d}t \, 
	\mathscr{L}_X^{n}(z)
	\right\}
	{\rm d}z.
\end{align}
Since, for each fixed \(z>0\), the mapping
\(t\in(0,\infty)\mapsto F_{Y_z}(t)\) in \eqref{notation} is the CDF of $Y_z$,
Proposition \ref{chave-prop} yields
\begin{align}\label{charact-2}
		\int_0^{\infty}
	F_{Y_z}(t^-) 
	\left[
	1
	-
	F_{Y_z}(t^-)
	\right] 
	{\rm d}t
	=
	\mathbb{E}[Y_z]G(Y_z),
\end{align}
where $G(Y_z)$ denotes the Gini coefficient of $Y_z$.

Finally, combining \eqref{charact-1} and \eqref{charact-2} completes the proof of the theorem.
\end{proof}

\begin{remark}\label{rem-main}
Observe that $\mathbb{E}[Y_z]$ and $G(Y_z)$ in Theorem \ref{main-theo} satisfy the following relations
	\begin{align*}
\mathbb{E}[Y_z]
=
{\mathbb{E}[X\exp\{-z X\}]\over \mathscr{L}_X(z)}
=
-
{\partial \log\left(\mathscr{L}_X(z)\right)\over \partial z}
	\end{align*}
	and
	\begin{align*}
	G(Y_z)
	=
	{
	\mathbb{E}\left[\vert X_1-X_2\vert \exp\{-z(X_1+X_2)\}\right]
	\over 
	2 	\mathbb{E}\left[X\exp\{-z(X)\}\right] \mathscr{L}_X(z)
}
=
A_1(z) A_2(z) \, G,
	\end{align*}
	where $G=G(X)$ is the Gini coefficient of $X$, $A_1$ is a weighted difference factor,
	\begin{align*}
	A_1(z)\equiv
	{\mathbb{E}\left[\vert X_1-X_2\vert \exp\{-z(X_1+X_2)\}\right]\over \mathbb{E}\vert X_1-X_2\vert} 
	\leqslant 1,
	\end{align*}
	because $\exp\{-z(X_1+X_2)\}$ is bounded above by 1,
	and $A_2$ is a mean adjustment factor,
	\begin{align*}
		A_2(z)\equiv {\mu\over \mathbb{E}[X\exp\{-zX\}] \mathscr{L}_X(z)}\geqslant 1,
	\end{align*}
	since, by the generalized Chebyshev sum inequality,
	$X$ and $\exp\{-z X\}$ are negatively correlated.
\end{remark}

By applying Theorem \ref{main-theo} together with Remark \ref{rem-main}, we have
\begin{corollary}\label{corollary-theo}
	Under the assumptions of Theorem \ref{main-theo},
	\[
	\mathbb{E}\left[\widehat{G}\right]
	=
	\left[
	-
	n
	\int_{0}^{\infty}
	{\partial \log\left(\mathscr{L}_X(z)\right)\over \partial z} \, 
	A_1(z) A_2(z)
	\mathscr{L}_X^{\,n}(z)\,
	\mathrm{d}z
	\right]
	G,
	\]
	where $G=G(X)$ is the Gini coefficient of $X$.
\end{corollary}

\subsection{Biased Gini estimator in Poisson populations}\label{Poisson-populations}

If $X\sim\mathrm{Poisson}(\lambda)$, $\lambda>0$, then the exponentially tilted variable
$Y_z$ is Poisson with mean $\lambda \exp\{-z\}$. By Proposition \ref{pro-2}, the
Gini coefficient of $X$ is
\[
G(X)=\exp\{-2z\}\big[I_0(2\lambda)+I_1(2\lambda)\big],
\]
where $I_0$ and $I_1$ denote the modified Bessel functions of the first kind.
By the same argument,
\[
G(Y_z)
=
\exp\{-2\lambda \exp\{-z\}\}
\big[I_0(2\lambda \exp\{-z\})+I_1(2\lambda \exp\{-z\})\big].
\]
Moreover, since $\mathbb{E}[Y_z]=\lambda \exp\{-z\}$ and
$\mathscr{L}_X(z)=\exp\{\lambda(\exp\{-z\}-1)\}$, the change of variables
$w=\exp\{-z\}$ in Theorem \ref{main-theo} yields
\[
\mathbb{E}\left[\widehat{G}\right]
=
n\lambda \exp\{-nz\}
\int_{0}^{1}
\exp\{(n-2)\lambda w\}
\big[I_0(2\lambda w)+I_1(2\lambda w)\big]
\,\mathrm{d}w.
\]
Although this integral has no closed form in terms of elementary or classical
special functions, the integrand is continuous on $[0,1]$; hence the integral is
finite and $\mathbb{E}\left[\widehat{G}\right]$ is well defined.

Moreover, since $1\leqslant I_0(2\lambda w)+I_1(2\lambda w)\leqslant I_0(2\lambda)+I_1(2\lambda)$, we get
\begin{align*}
	{\exp\{-2\lambda\}-\exp\{-n\lambda\}\over (1-{2\over n})}
	\leqslant 
	\mathbb{E}\left[\widehat{G}\right]
	\leqslant
	{1-\exp\{-(n-2)\lambda\}\over (1-{2\over n})}
	\, G,
	\quad 
	n\geqslant 1, n\neq 2,
\end{align*}
and since
$
	\int_{0}^{1}
	\bigl[I_0(2\lambda w)+I_1(2\lambda w)\bigr]\,
	\mathrm{d}w
	=
	[{I_0(2\lambda)+I_1(2\lambda)-1}]/({2\lambda}),
$
we have
\begin{align*}
	\mathbb{E}\left[\widehat{G}\right]
	=
G-\exp\{-2\lambda\},
\quad 
n=2.
\end{align*}

In other words, the estimator 
$\widehat{G}$
is biased for Poisson populations,
with bias
\begin{align}\label{bias_poisson}
\text{Bias}(\widehat{G},G)
=
	n\lambda\exp\{-n\lambda\}
\int_{0}^{1}
\exp\{(n-2)\lambda w\}
[I_0(2\lambda w)+I_1(2\lambda w)]
{\rm d}w
-
\exp\{-2\lambda\}[I_0(2\lambda)+I_1(2\lambda)]
\end{align}
satisfying 
\begin{align*}
	{\exp\{-2\lambda\}-\exp\{-n\lambda\}\over (1-{2\over n})}
-
G
	\leqslant 
\text{Bias}(\widehat{G},G)
	\leqslant
	\left[
	{1-\exp\{-(n-2)\lambda\}\over (1-{2\over n})}
	-
	1
	\right]
	G,
	\quad 
	n\geqslant 1, n\neq 2,
\end{align*}
and
\begin{align*}
\text{Bias}(\widehat{G},G)
=
-\exp\{-2\lambda\},
\quad 
n=2.
\end{align*}

\subsection{Biased Gini estimator in geometric populations}\label{Geometric-populations}

If $X\sim\mathrm{Geometric}(p)$, $p\in(0,1)$, on $\{0,1,\ldots\}$, then the exponentially
tilted variable $Y_z$ is geometric with parameter
$1-(1-p)\exp\{-z\}$. By Proposition \ref{pro-2},
\[
G(X)=\frac{1}{2-p},
\quad
G(Y_z)=\frac{1}{1+(1-p)\exp\{-z\}}.
\]
Moreover, since
\[
\mathbb{E}[Y_z]=\frac{(1-p)e^{-z}}{1-(1-p)\exp\{-z\}}
\quad\text{and}\quad
\mathscr{L}_X(z)=\frac{p}{1-(1-p)\exp\{-z\}},
\]
the change of variables $w=(1-p)\exp\{-z\}$ yields the representation of
$\mathbb{E}[\widehat{G}]$ in Theorem \ref{main-theo} as
\begin{align*}
\mathbb{E}\left[\widehat{G}\right]
=
n p^n
\int_{0}^{1-p}
{1\over (1-w)^{n+1} (1+w)} 
\,
\mathrm{d}w
=
{1\over 2}
\left[
 _2F_1\left(1,n;n+1;{p\over 2}\right)- p^{n} \, _2F_1\left(1,n;n+1;{1\over 2}\right)
\right],
\end{align*}
where in the last equality we have used the standard integral representation of the Gauss hypergeometric function \citep[see Chapter 15, Item 15.3. in][]{Abramowitz1972}:
\begin{equation*}
	{}_2F_1(a,b;c;z)
	=
	\frac{\Gamma(c)}{\Gamma(b)\Gamma(c-b)}
	\int_{0}^{1}
	t^{b-1}(1-t)^{c-b-1}(1-zt)^{-a}\,{\rm d}t,
	\quad
	c>b>0.
\end{equation*}

Since, for $0\leqslant w\leqslant 1-p$, it holds that $G\leqslant 1/(1+w)\leqslant 1$, it follows that
\begin{align*}
	(1-p^n)\, G
	=
	\left[
		n p^n
	\int_{0}^{1-p}
	{1\over (1-w)^{n+1}} 
	\,
	\mathrm{d}w
	\right]
	G
	\leqslant
	\mathbb{E}\left[\widehat{G}\right]
	\leqslant
	n p^n
	\int_{0}^{1-p}
	{1\over (1-w)^{n+1}} 
	\,
	\mathrm{d}w
	=
	{1-p^n}.
\end{align*}

In other words, the estimator $\widehat{G}$ is biased for Geometric populations, with bias given by
\begin{align}\label{bias_geometric}
\text{Bias}(\widehat{G},G)
=
{1\over 2}
\left[
_2F_1\left(1,n;n+1;{p\over 2}\right)- p^{n} \, _2F_1\left(1,n;n+1;{1\over 2}\right)
\right]
-
{1\over 2-p}
\end{align}
which satisfies
\begin{align*}
	-p^n G
	\leqslant
	\text{Bias}(\widehat{G},G)
	\leqslant
	1	-p^n-G.
\end{align*}


\section{Characterization of the gamma family}\label{Characterization}

The next result presents a characterization of the gamma family which, to the best of our knowledge, has not previously appeared in the literature.

\begin{theorem}\label{main-theorem}
	Let $X$ be a nonnegative, absolutely continuous random variable belonging to a
	scale family, with density $f_X$. Assume that for every $z > 0$ there
	exists a function $\xi(z)>0$ such that, for almost every $t>0$,
	\begin{equation}
		\label{eq:char-gamma}
		\exp\{-zt\} f_X(t)
		=
		\frac{\mathscr{L}_X(z)}{\xi(z)}\,
		f_X\left(\frac{t}{\xi(z)}\right),
	\end{equation}
	where $\mathscr{L}_{X}(z)=\mathbb{E}[\exp\{-zX\}]$ denotes the Laplace transform of $X$, and $F_X$ is its corresponding CDF.
	Then $X \sim \text{Gamma}(\alpha,\lambda)$, for some $\alpha>0$, $\lambda>0$ (gamma distribution). More precisely,
	\[
	f_X(x)={\lambda^\alpha\over\Gamma(\alpha)}\, x^{\alpha-1}\exp\{-\lambda x\}, \quad x>0,
	\]
	and necessarily
	\[
	\xi(z)=\frac{\lambda}{\lambda+z}.
	\]
	
	Conversely, if $X \sim \text{Gamma}(\alpha,\lambda)$, then Identity
	\eqref{eq:char-gamma} holds for all $z > 0$.
\end{theorem}
\begin{proof}
Since $X$ belongs to a scale family, there exists a baseline density $g$ such that
\begin{align}\label{id-1}
	f_X(x)=\frac{1}{\theta}\, g\left(\frac{x}{\theta}\right),
\end{align}
for some $\theta>0$. Without loss of generality, assume that $\theta=1$. Equation \eqref{eq:char-gamma} can therefore be rewritten as
\begin{align*}
	\exp\{-zt\}\, g\left({t}\right)
	=
	\frac{\mathscr{L}_X(z)}{\xi(z)}\,
	g\left(\frac{t}{\xi(z)}\right).
\end{align*}
Taking logarithms and differentiating with respect to $t$ yields
\begin{align*}
	-z+
	\left[\log(g({t}))\right]'
	=
	\frac{1}{\xi(z)}
	\left[\log\left(g\left(\frac{t}{\xi(z)}\right)\right)\right]'.
\end{align*}
Differentiating once more with respect to $t$, we obtain
\begin{align*}
	\left[\log(g\left({t}\right))\right]''
	=
	\frac{1}{\xi^{2}(z)}
	\left[
	\log\left(g\left(\frac{t}{\xi(z)}\right)\right)
	\right]''.
\end{align*}
Letting $x=t$ and $y=t/\xi(z)$, the above equation becomes
\begin{align*}
	x^{2}\bigl[\log g(x)\bigr]''
	=
	y^{2}\bigl[\log g(y)\bigr]'',
\end{align*}
for all $x,y>0$. Hence, the quantity $x^{2}[\log g(x)]''$ must be constant. That is, there exists $C\in\mathbb{R}$ such that
\begin{align*}
	x^{2}\bigl[\log g(x)\bigr]''=C.
\end{align*}
Solving this differential equation yields
\begin{align}\label{id-2}
	g(x)=\exp\{K\}x^{-C}\exp\{Dx\},
\end{align}
for some constants $K$, $C$, and $D$. Since $g$ is a probability density on $(0,\infty)$, it follows that $-C>-1$ and $D<0$. Setting $\alpha=1-C$ and $\lambda=-D$, normalization together with \eqref{id-1} (with $\theta=1$) and \eqref{id-2} implies that $X\sim\mathrm{Gamma}(\alpha,\lambda)$.

The proof of the reciprocal implication is immediate and therefore omitted.
\end{proof}

\begin{corollary}\label{cor-main}
	Under the assumptions of Theorem \ref{eq:char-gamma},
	\[
	\frac{\mathbb{E}\left[\exp\{-zX\} \mathds{1}_{\{X\leqslant t\}}\right]}
	{\mathscr{L}_X(z)}
	=
	F_X\left(\frac{t}{\xi(z)}\right),
	\quad t>0;
	\qquad
	\xi(z)=\frac{\lambda}{\lambda+z},
	\quad z>0.
	\]
\end{corollary}

\subsection{Unbiased Gini estimator in gamma populations}\label{Unbiasedness}

	Let $X$ be a nonnegative, absolutely continuous random variable belonging to a
scale family, with density $f_X$, satisfying \eqref{eq:char-gamma}. Using the notation $F_{Y_z}$ introduced in \eqref{notation}, Theorem \ref{main-theorem} implies that $X \sim \text{Gamma}(\alpha,\lambda)$ and Corollary \ref{cor-main} yields
\begin{align*}
	F_{Y_z}(t)
	=
	F_X\left(\frac{t}{\xi(z)}\right)
	=
	F_{\xi(z) X}(t),
	\quad
	\xi(z)=\frac{\lambda}{\lambda+z}.
\end{align*}
In other words, 
\begin{align*}
Y_z\stackrel{d}{=}\xi(z)X\sim \text{Gamma}(\alpha,\lambda+z).
\end{align*}
Hence, $E[Y_z]=\alpha/(\lambda+z)$, $\mathscr{L}_X(z)=\lambda^\alpha/(\lambda+z)^\alpha$, and since the Gini coefficient is scale invariant, 
\begin{align*}
G(Y_z)=G(X)=G=\frac{\Gamma(2\alpha+1)}{2^{2\alpha}\Gamma(\alpha+1)^2},
\end{align*}
where the last identity is well known in the literature and can be easily obtained from Proposition \ref{chave-prop}.
Applying Theorem \ref{main-theo}, we obtain
\begin{align*}
\mathbb{E}\left[\widehat{G}\right]
=
n
\int_{0}^{\infty}
\mathbb{E}[Y_z]
G(Y_z)
\mathscr{L}_X^{\,n}(z)\,
\mathrm{d}z
=
G
\left[
n\alpha\lambda^{n\alpha}
\int_{0}^{\infty}
{1\over(\lambda+z)^{n\alpha+1}}\,
\mathrm{d}z
\right]
=
G.
\end{align*}
That is, under gamma populations the Gini estimator $\widehat{G}$ is unbiased for the Gini coefficient $G$, as previously shown in \cite{Baydil2025} and \cite{VilaSaulomixture2025,VilaSaulomth2025}.

%
%
%
%
%
%

\section{Simulation study}\label{sec:sim}

This section reports a Monte Carlo study assessing the finite-sample performance of the Gini
estimator in~\eqref{def-G-estimator} (uncorrected) and its bias-corrected version obtained by subtracting
an estimate of the finite-sample bias. The correction is implemented in a plug-in fashion using
maximum likelihood (ML) estimation of the distributional parameter.

Let $X_1,\ldots,X_n$ be an i.i.d. sample from a Poisson or Geometric distribution.
The uncorrected estimator is given in \eqref{def-G-estimator}, that is,
\begin{align*}
	\widehat{G}
	=
	\begin{cases}
	\dfrac{ {1\over \binom{n}{2}} \sum_{1\leqslant i<j\leqslant n}\vert X_i-X_j\vert}{2\overline{X} },
& \text{if} \  \sum_{k=1}^{n}X_k>0 \ \text{a.s.},
\\[0,3cm]
0, & \text{if} \  \sum_{k=1}^{n}X_k=0 \ \text{a.s.},
	\end{cases}
\end{align*}
where
$
\overline{X}=\sum_{k=1}^{n} X_k / n
$
denotes the sample mean. The corrected estimator is defined as
\begin{equation}\label{eq:gini_hat_corr}
\widehat G^{\,c}
=
\widehat G - \widehat{\mathrm{Bias}}_n(\widehat\theta),
\end{equation}
where $\widehat\theta$ is the ML estimator of the model parameter and $\widehat{\mathrm{Bias}}_n(\widehat\theta)$ (Equations \eqref{bias_poisson} and \eqref{bias_geometric})
is obtained by evaluating the analytical bias expression of $\widehat G$ at $\widehat\theta$ and the
given sample size $n$.

For the Poisson model, we consider $X \sim \mathrm{Poisson}(\lambda)$ with
$\lambda \in \{0.5,1,2,5,10\}$. For each Monte Carlo replication and each pair $(n,\lambda)$,
the ML estimator of $\lambda$ is $\widehat\lambda=\bar X$. The plug-in bias correction is computed using
the analytical expression in Eq.~\eqref{bias_poisson}.

For the geometric model, we generate $X \sim \mathrm{Geometric}(p)$ on $\{0,1,\ldots\}$ with
$p \in \{0.1,0.2,0.4,0.6,0.8\}$. For each replication and each pair $(n,p)$, the ML estimator is
$\widehat p = 1/(1+\bar X)$ and the bias correction is obtained from Eq.~\eqref{bias_geometric}.

Sample sizes are set to $n \in \{25,50,75,100\}$, and each configuration is replicated
$R$ times, with $R$ sufficiently large to ensure stable Monte Carlo estimates
(in our implementation, $R=10{,}000$). For each replication $r$, we compute
$(\widehat G_r,\widehat G_{c,r})$.

Performance is summarized using the relative bias and the root mean squared error (RMSE),
defined respectively as
\[
\mathrm{RelBias}(\widehat G_e)=
\frac{1}{R}\sum_{r=1}^R\frac{\widehat G_{e,r}-G}{G},
\qquad
\mathrm{RMSE}(\widehat G_e)=
\left(\frac{1}{R}\sum_{r=1}^R(\widehat G_{e,r}-G)^2\right)^{1/2},
\]
where $\widehat G_e$ denotes either $\widehat G$ or $\widehat G_c$ and $G$ is the true Gini
coefficient under the data-generating model.

Figure~\ref{fig:mc_poisson} reports the Monte Carlo results for the Poisson model. Panels (a) and (b)
display the relative bias and RMSE as functions of the sample size $n$, for fixed values of
$\lambda$, while panels (c) and (d) show the same measures as functions of $\lambda$, for fixed
sample sizes. From this figure, we observe that the bias-corrected estimator displays
reduced relative bias across all configurations. Regarding efficiency, the RMSE decreases with increasing $n$ for both estimators, with quite similar performances.

\begin{figure}[!ht]
\centering
\subfigure[Relative bias vs $n$ (Poisson)]{
  \includegraphics[width=0.47\textwidth]{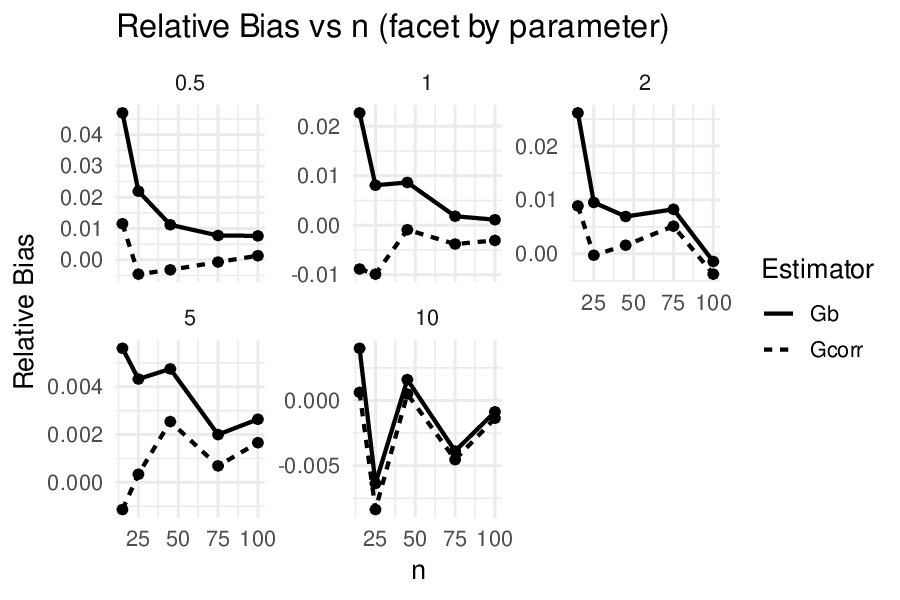}
}
\subfigure[RMSE vs $n$ (Poisson)]{
  \includegraphics[width=0.47\textwidth]{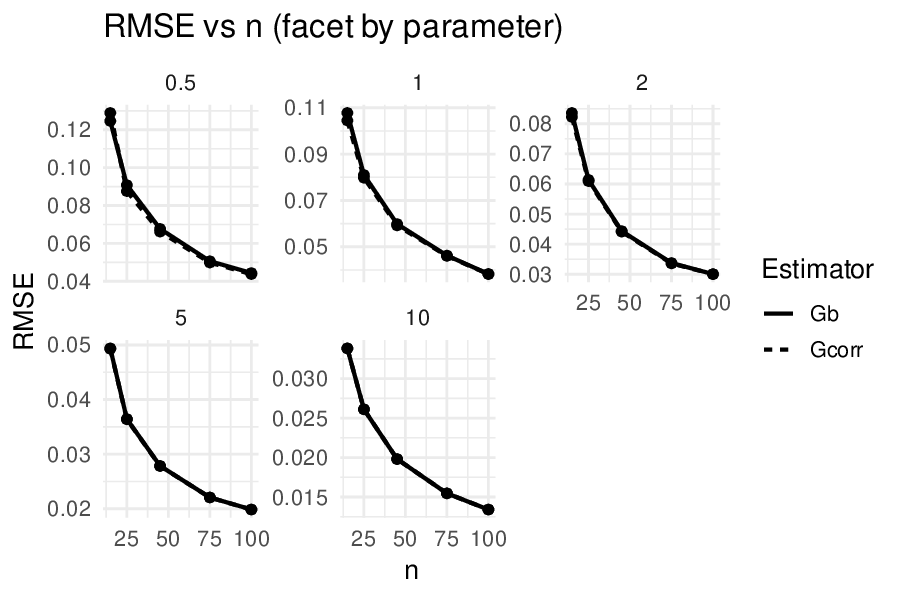}
}\\
\subfigure[Relative bias vs $\lambda$ (Poisson)]{
  \includegraphics[width=0.47\textwidth]{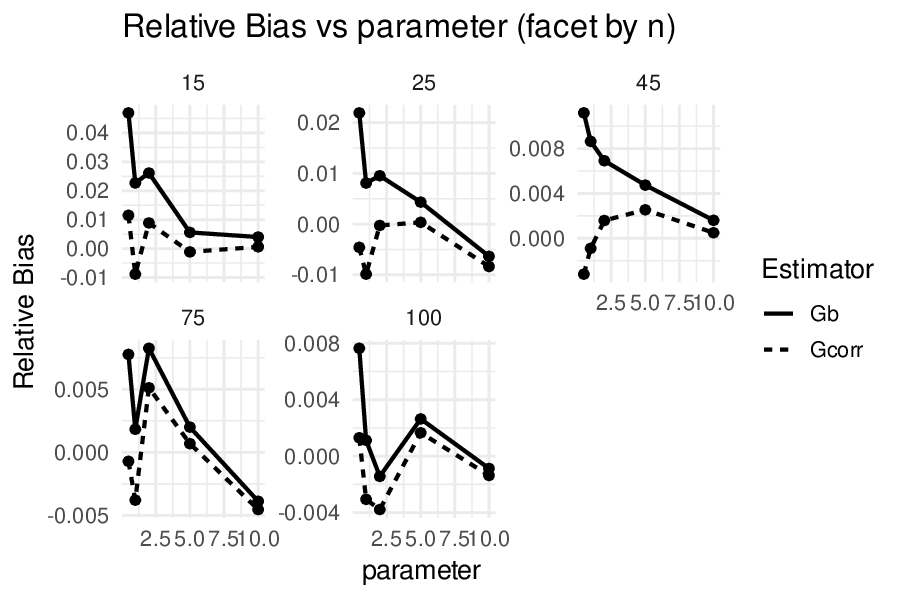}
}
\subfigure[RMSE vs $\lambda$ (Poisson)]{
  \includegraphics[width=0.47\textwidth]{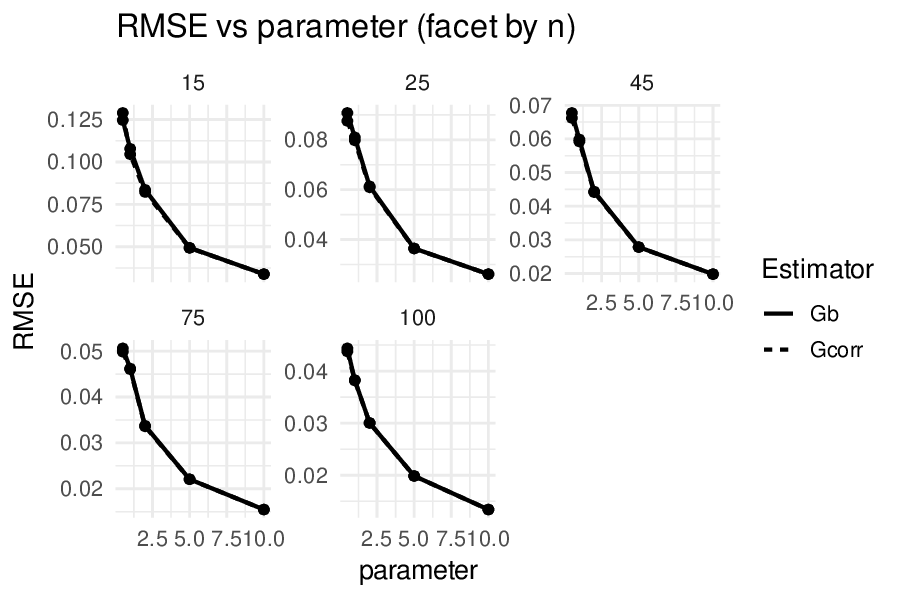}
}
\caption{Monte Carlo results for the Poisson model: relative bias and RMSE for the uncorrected estimator (Gb, Eq.~\eqref{def-G-estimator}) and the bias-corrected version (Gcorr).}
\label{fig:mc_poisson}
\end{figure}

Figure~\ref{fig:mc_geometric} presents the corresponding results for the geometric model. Panels (a) and (b)
show relative bias and RMSE as functions of $n$, whereas panels (c) and (d) depict these measures
as functions of the success probability $p$. As in the Poisson case, the bias-corrected estimator displays
reduced relative bias across all configurations, with similar RMSE performances.

\begin{figure}[!ht]
\centering
\subfigure[Relative bias vs $n$ (Geometric)]{
  \includegraphics[width=0.47\textwidth]{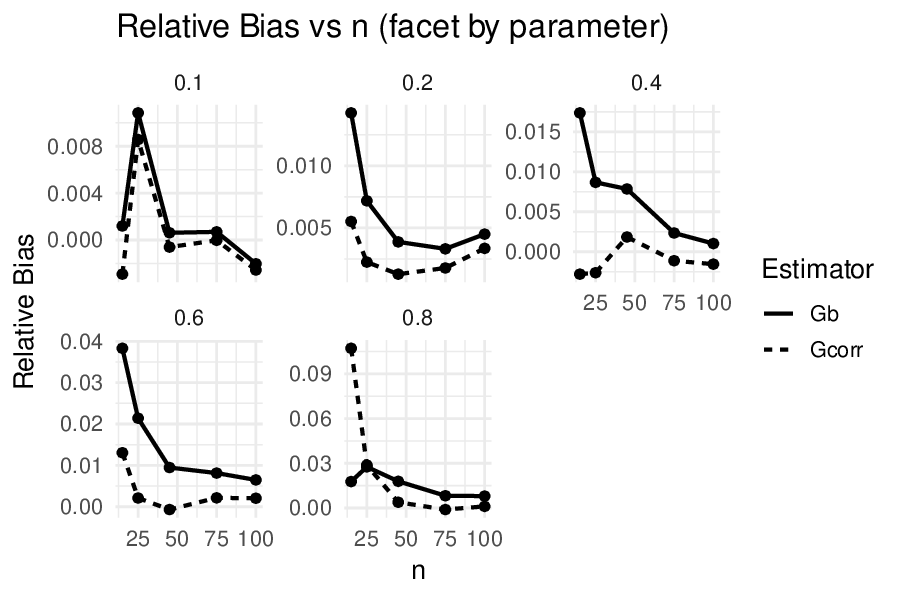}
}
\subfigure[RMSE vs $n$ (Geometric)]{
  \includegraphics[width=0.47\textwidth]{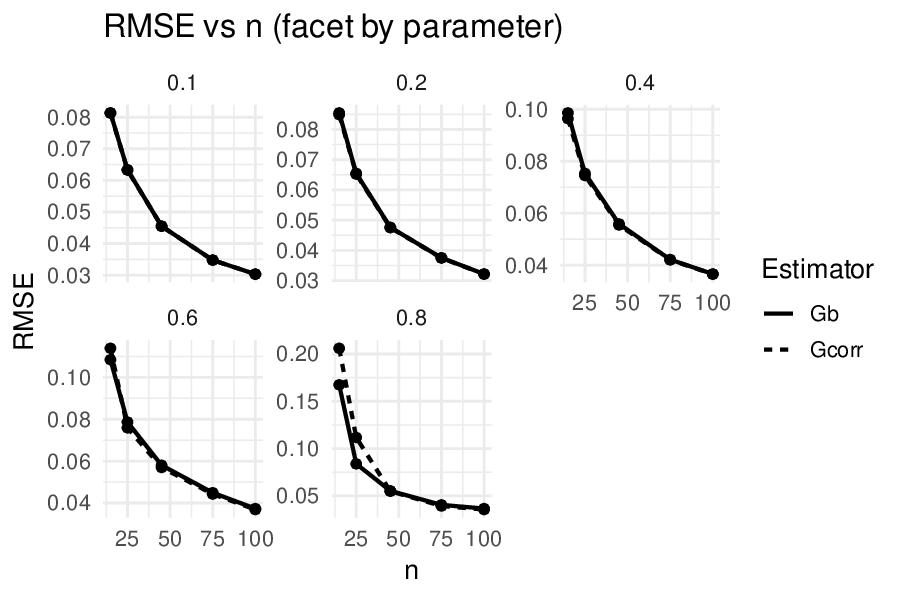}
}\\
\subfigure[Relative bias vs $p$ (Geometric)]{
  \includegraphics[width=0.47\textwidth]{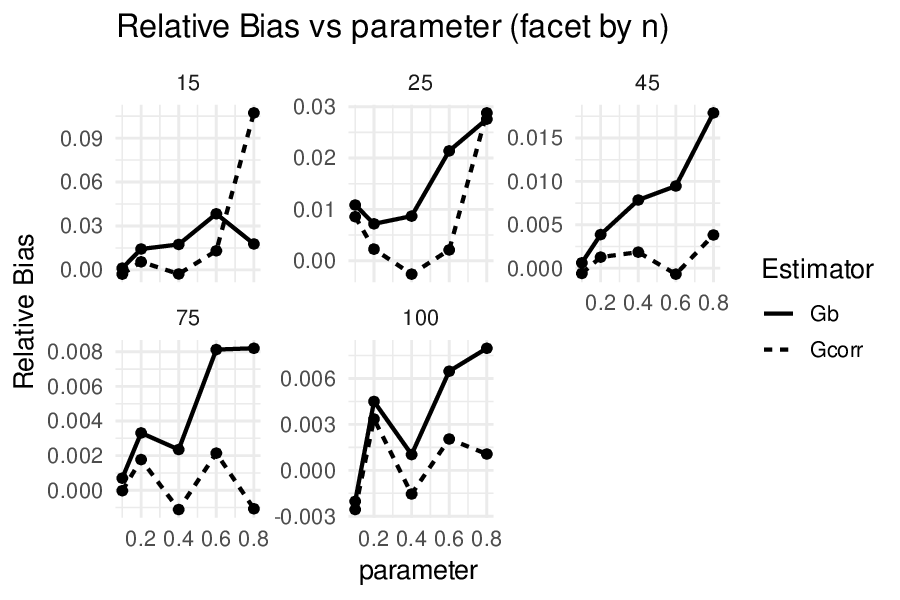}
}
\subfigure[RMSE vs $p$ (Geometric)]{
  \includegraphics[width=0.47\textwidth]{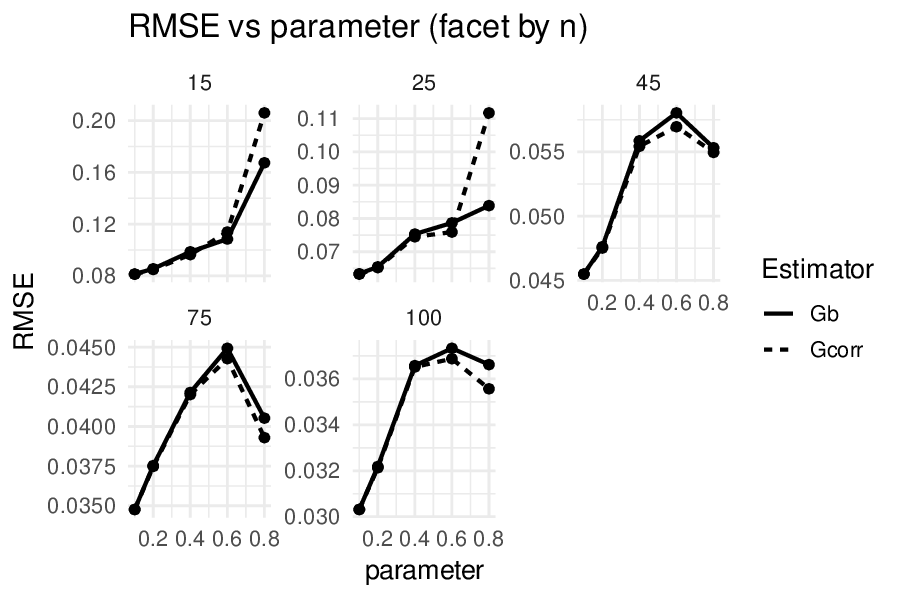}
}
\caption{Monte Carlo results for the geometric model: relative bias and RMSE for the uncorrected estimator (Gb, Eq.~\eqref{def-G-estimator}) and the bias-corrected version (Gcorr).}
\label{fig:mc_geometric}
\end{figure}

\section{Concluding remarks}\label{concluding_remarks}

In this paper, we derived a general representation for the expectation of the Gini coefficient estimator based on the Laplace transform of the underlying distribution and the Gini coefficient of its exponentially tilted version. This formulation yields a new characterization of the gamma family within the class of nonnegative scale families via a stability property under exponential tilting. As direct applications, we showed that the Gini estimator is biased for Poisson and geometric populations, while remaining unbiased for gamma models. By exploiting the derived bias expressions, we proposed plug-in bias-corrected estimators and assessed their finite-sample performance through Monte Carlo simulations, which confirmed substantial improvements over the uncorrected estimator.

\section*{Acknowledgments}

The research was supported in part by CNPq and CAPES grants from the Brazilian government. Helton Saulo gratefully acknowledges financial support from: (ii) the University of Brasília; and (ii) the Brazilian National Council for Scientific and Technological Development (CNPq), through Grant 304716/2023-5.

\section*{Declaration}

There are no conflicts of interest to disclose.




\begin{thebibliography}{}

	
	\bibitem[Aaberge, 2000]{Aaberge2000}
	Aaberge, R. (2000).
	\newblock Characterizations of lorenz curves and income distributions.
	\newblock {\em Social Choice and Welfare}, 17:639--653.
	\newblock \url{https://link.springer.com/article/10.1007/s003550000046}.
	
		\bibitem[Abramowitz and Stegun, 1972]{Abramowitz1972}
	Abramowitz, M. and Stegun, I. A. (1972).
	{\it Handbook of Mathematical Functions with Formulas, Graphs, and Mathematical Tables}.
	Dover, New York.
	
	\bibitem[Baydil et~al., 2025]{Baydil2025}
	Baydil, B., de~la Peña, V.~H., Zou, H., and Yao, H. (2025).
	\newblock Unbiased estimation of the gini coefficient.
	\newblock {\em Statistics and Probability Letters}, 222:110376.
	\newblock \url{https://doi.org/10.1016/j.spl.2025.110376}.
	
	\bibitem[Butler, 2007]{Butler2007}
	Butler, R., editor (2007).
	\newblock {\em Saddlepoint Approximations with Applications}.
	\newblock Cambridge University Press, Cambridge.
	
	\bibitem[De la Peña et al., 2025]{Pena2025} 
	De la Pe\~na, V. H., Yao, H., and Zou, H. (2025).
	A scalable formula for the moments of a family of self-normalized statistics.
	Available at \url{https://arxiv.org/pdf/2509.14428}.
	

	\bibitem[Deltas, 2003]{Deltas2003}
	Deltas, G. (2003).
	\newblock The small-sample bias of the gini coefficient: results and
	implications for empirical research.
	\newblock {\em Review of Economics and Statistics}, 85:226--234.
	
	\bibitem[Donaldson and Weymark, 1980]{Donaldson1980}
	Donaldson, D. and Weymark, J. (1980).
	\newblock Single parameter generalization of the gini indices of inequality.
	\newblock {\em Journal of Economic Theory}, 22:67--86.
	
	\bibitem[Gini, 1936]{Gini1936}
	Gini, C. (1936).
	\newblock On the measure of concentration with special reference to income and
	statistics.
	\newblock {\em Colorado College Publication, General Series}, (208):73--79.
	
	\bibitem[Kakwani, 1980]{Kakwani1980}
	Kakwani, N. (1980).
	\newblock On a class poverty measures.
	\newblock {\em Econometrica}, 48:437--446.
	\newblock \url{https://www.jstor.org/stable/1911106?origin=crossref}.
	
	\bibitem[Vila et~al., 2024]{VILA2024110032}
	Vila, R., Balakrishnan, N., and Saulo, H. (2024).
	\newblock An upper bound and a characterization for gini's mean difference
	based on correlated random variables.
	\newblock {\em Statistics and Probability Letters}, 207:110032.
	
	\bibitem[Vila and Saulo, 2025a]{VilaSaulomixture2025}
	Vila, R. and Saulo, H. (2025a).
	\newblock Bias in {G}ini coefficient estimation for gamma mixture populations.
	\newblock {\em Statistical Papers}, 66:1--18.
	\newblock \url{https://doi.org/10.1007/s00362-025-01768-w}.
	
	\bibitem[Vila and Saulo, 2025b]{VilaSaulomth2025}
	Vila, R. and Saulo, H. (2025b).
	\newblock The mth gini index estimator: Unbiasedness for gamma populations.
	\newblock Accepted for publication in The Journal of Economic Inequality.
	Available at \url{https://arxiv.org/abs/2504.19381}.
	
	\bibitem[Yitzhaki, 1983]{Yitzhaki1983}
	Yitzhaki, S. (1983).
	\newblock On a extension of gini inequality index.
	\newblock {\em International Economic Review}, 24:617--628.
	
	
	
\end{thebibliography}

\end{document}